\documentclass[journal]{IEEEtran}

\usepackage{amsmath}
\usepackage{amssymb}
\usepackage{amsthm}
\usepackage{graphicx}
\usepackage{epstopdf}
\usepackage{cite}
\usepackage{color,array}

\newtheorem{proposition}{Proposition}
\begin{document}

\title{A Closed-Form Noise-Sensitivity Asymmetry for Causal Branch Selection in Minimal-Array TDoA Localization}

\author{Abeer~Nasir~Chaudhry,
        Salman~Liaquat,~\IEEEmembership{Senior~Member,~IEEE,}
        and~Hasan~Saeed~Mir,~\IEEEmembership{Senior~Member,~IEEE}%
\thanks{This work is under review as a Correspondence in the IEEE Transactions on Aerospace and Electronic Systems. This work has been submitted to the IEEE for possible publication. Copyright may be transferred without notice, after which this version may no longer be accessible. (\textit{Corresponding author: Hasan Mir.})}%
\thanks{A. N. Chaudhry and H. S. Mir are with the American University of Sharjah, Sharjah, United Arab Emirates (e-mail: b00085496@aus.edu; hmir@aus.edu).}%
\thanks{S. Liaquat is with the School of Electrical and Electronic Engineering, Universiti Sains Malaysia, Penang, Malaysia (e-mail: salmanliaquat@ieee.org).}%
\thanks{Color versions of one or more figures in this article are available online at http://ieeexplore.ieee.org.}}

\markboth{IEEE Transactions on Aerospace and Electronic Systems}{Chaudhry \MakeLowercase{\textit{et al.}}: Closed-Form Noise-Sensitivity Asymmetry for Causal Branch Selection in Minimal-Array TDoA Localization}

\maketitle

\begin{abstract}
Minimal-array time-difference-of-arrival (TDoA) localization with three planar receivers reduces to a scalar quadratic whose two roots can both be feasible target positions, leaving a branch-selection ambiguity intrinsic to the minimal geometry. Because a minimal array is the cheapest, most deployable passive configuration, this ambiguity is conventionally broken only by adding a fourth receiver or an angle sensor, sacrificing the very minimality that motivates the array. This article resolves it from the measurements alone. Implicit differentiation shows that the two roots share an identical sensitivity denominator, so classical root conditioning cannot separate them and the entire asymmetry resides in the numerator. A closed-form analysis then yields an exact sign condition for the asymmetry, and over the feasible interior the more sensitive root is the outer, physical one, away from two algebraic degeneracy loci. A causal, constant-memory selector smooths a per-root variability statistic, whose expectation is shown to be proportional to the per-root sensitivity in the noise-dominated regime, and selects the larger one. Across a dimensionless receiver atlas the physical root is the more sensitive at a median of 85\% of two-feasible-root operating points (median sensitivity ratio 16), and, because the more sensitive root is the outer one, a simple outer-root rule attains near-0.99 branch-selection accuracy on this interior, matched online by a causal, constant-memory realization. The decisive empirical finding is that smoothness- and continuity-based disambiguation, the natural alternatives, invert under timing noise and fall below chance.
\end{abstract}


\IEEEpeerreviewmaketitle

\section{INTRODUCTION}
Time-difference-of-arrival (TDoA) localization estimates an emitter position without the antenna complexity of angle-of-arrival sensing. In two dimensions with three non-collinear receivers, the analytical formulation reduces, after a rigid transformation of the geometry to a canonical frame, to a scalar quadratic in a reference range. Depending on the measurement and geometry, the quadratic yields one feasible target position, two feasible positions, or none. The two-feasible-root case is the focus of this article. One branch is the physical target, and the other is the second geometric preimage of the same reference-range quadratic. Both are exact roots. The second arises because the inverse mapping from range-difference space to Cartesian space is non-injective, not from any extraneous solution introduced by squaring. The ambiguity is intrinsic to the minimal three-receiver planar geometry, which admits dual solutions over a range of arrival-time differences, whereas four non-collinear receivers yield a unique solution \cite{spencer2007twod}. From a single measurement the two candidates carry no label.

Reference-free disambiguation has been approached in three ways. One re-solves the geometry across multiple epochs to eliminate the mirror by geometric convergence \cite{huo2021mirror}, which needs an over-determined multi-epoch solve, and a related multi-epoch scheme fuses time-difference observations over a satellite triad with Kalman and Gaussian-mixture filtering to resolve time-difference ambiguity \cite{zhang2025obsfilter}. Another enumerates and associates algebraic candidates across receiver subsets \cite{flood2024association}. A third augments the observation with an angle measurement \cite{dogancay2023conic,ho2023aoatdoa}, which removes the planar sign ambiguity at the cost of an additional sensing modality. The closed-form solver itself has been refined repeatedly. Recent work includes the classical hyperbolic estimator \cite{chanho1994}, the complete analytical solution \cite{hubacek2022}, asymptotically optimal range-difference estimation \cite{sun2023tdoaspeed}, a modified-polar-representation projection that unifies the near and far field \cite{sun2024mpr}, and an eigenvalue formulation that exposes the algebraic solution multiplicity \cite{larsson2025eigen}. In parallel, the planar TDoA-map image and bifurcation curve are characterized in \cite{compagnoni2014geometry,compagnoni2014bifurcation}, and optimal sensor placement is designed for the estimation accuracy of an already-resolved solution \cite{zhang2025osp}. Each returns the source position accurately once the branch is fixed, but none identifies the physical branch from the per-sample noise response of the analytical roots, causally and without re-solving the geometry across epochs or enumerating candidate trajectories, nor explains why such a criterion can exist. This article closes that gap with a single observation. The two roots respond differently to timing-measurement perturbations, and the physical root is the more sensitive one.

The contribution is not the quadratic solver or the TDoA-map geometry, which are treated as the mathematical foundation \cite{compagnoni2014geometry,compagnoni2014bifurcation,hubacek2022}. The contributions are:
\begin{enumerate}\renewcommand{\labelenumi}{(\roman{enumi})}
    \item a closed-form branch-sensitivity result that shows the physical root to be the more measurement-sensitive of the two, unlike the normalized conditioning of the quadratic root, and \emph{characterizes} this polarity by an exact closed-form sign condition, the more sensitive branch being the outer root away from two algebraic degeneracy loci.
    \item the observation that this sign condition justifies the simple outer-root (larger-range) rule as the branch selector on the admissible interior, together with a causal, constant-memory realization of that criterion, termed temporal-variability solution disambiguation (TVSD), that recovers the same selection online from the measurement stream, with an analysis showing its statistic is, in expectation, the per-root sensitivity under a noise-dominance condition.
    \item a sensitivity-polarity corollary stating that in the two-feasible-root interior, wherever the physical root is the higher-variance branch (the large majority of operating points, Sec.~\ref{sec:sens}), every smoothness-based selector mislabels it under noise. 
\end{enumerate}

\section{ANALYTICAL TDOA ROOT MODEL}
\emph{Notation.} Boldface denotes vectors and matrices, $[n]$ a discrete sample index, $\|\cdot\|$ the Euclidean norm, $\hat{(\cdot)}$ an estimate, $(\cdot)^T$ transpose, and subscripts index receivers or analytical roots. 
Let the receiver positions be $\mathbf{p}_i\in\mathbb{R}^2$, $i\in\{0,1,2\}$, with $\mathbf{p}_0$ the reference. At sample $n$ the target is $\mathbf{x}[n]$ and the propagation range to receiver $i$ is $r_i=\|\mathbf{x}[n]-\mathbf{p}_i\|$. Expressed in range units, the measurement vector is
\begin{equation}
  \mathbf{z}[n]=
  \begin{bmatrix} z_1[n]\\ z_2[n] \end{bmatrix}
  =
  \begin{bmatrix} r_1-r_0\\ r_2-r_0 \end{bmatrix}
  +\boldsymbol{\nu}[n],
  \label{eq:measurement_vector}
\end{equation}
where $\boldsymbol{\nu}[n]$ collects timing and preprocessing errors as range errors.

A rigid transformation with orthogonal $\mathbf{Q}$, $\mathbf{y}=\mathbf{Q}^{T}(\mathbf{x}-\mathbf{p}_0)$, $\mathbf{a}_i=\mathbf{Q}^{T}(\mathbf{p}_i-\mathbf{p}_0)$, places the reference at the origin and the receivers at $\mathbf{a}_0=\mathbf{0}$, $\mathbf{a}_1=[a,0]^{T}$, $\mathbf{a}_2=[b,c]^{T}$ with $a>0$, $c\neq0$, and defines the canonical reference range $\rho[n]=\|\mathbf{y}[n]\|=r_0$. Squaring the canonical TDoA equations $\|\mathbf{y}-\mathbf{a}_i\|-\rho=z_i$ and substituting $\rho^2=\mathbf{y}^{T}\mathbf{y}$ linearizes the geometry to an affine model $\mathbf{y}(\rho)=\boldsymbol{\alpha}+\boldsymbol{\beta}\rho$, whose coefficients depend on $(a,b,c)$ and $\mathbf{z}$ (the transformation and the coefficients are derived in the supplementary material, following \cite{hubacek2022}). Enforcing $\rho^2=\|\boldsymbol{\alpha}+\boldsymbol{\beta}\rho\|^2$ yields the scalar quadratic
\begin{equation}
  \lambda_2\rho^2+\lambda_1\rho+\lambda_0=0,
  \label{eq:quadratic}
\end{equation}
with $\lambda_2=\boldsymbol{\beta}^{T}\boldsymbol{\beta}-1$, $\lambda_1=2\boldsymbol{\alpha}^{T}\boldsymbol{\beta}$, $\lambda_0=\boldsymbol{\alpha}^{T}\boldsymbol{\alpha}$, and discriminant $\Delta_\rho=\lambda_1^2-4\lambda_2\lambda_0$. When $\lambda_2\neq0$ and $\Delta_\rho\geq0$ the analytical roots are
\begin{equation}
  \rho_m=\frac{-\lambda_1+(-1)^{m+1}\sqrt{\Delta_\rho}}{2\lambda_2},
  \qquad m\in\{1,2\},
  \label{eq:rho_roots}
\end{equation}
each mapping back to a candidate location $\mathbf{x}_m=\mathbf{p}_0+\mathbf{Q}(\boldsymbol{\alpha}+\boldsymbol{\beta}\rho_m)$. The per-sample root count classifies the inverse mapping, and the branch selector below is applied only when two distinct feasible (positive, real) roots are present.

\section{BRANCH SENSITIVITY ASYMMETRY}
\label{sec:sens}
The selection statistic is justified only if the two roots respond differently to measurement perturbations. Treating each root in \eqref{eq:rho_roots} as an implicit function of $\mathbf{z}$ through $\lambda_2(\mathbf{z})$, $\lambda_1(\mathbf{z})$, $\lambda_0(\mathbf{z})$, and differentiating \eqref{eq:quadratic} with respect to a component $z_j$ gives
\begin{equation}
  \frac{\partial\rho_m}{\partial z_j}=-\frac{g_j(\rho_m)}{2\lambda_2\rho_m+\lambda_1},
  \quad
  g_j(\rho)=\tfrac{\partial\lambda_2}{\partial z_j}\rho^2+\tfrac{\partial\lambda_1}{\partial z_j}\rho+\tfrac{\partial\lambda_0}{\partial z_j}.
  \label{eq:implicit_sens}
\end{equation}
By \eqref{eq:rho_roots}, the denominator evaluates to $2\lambda_2\rho_m+\lambda_1=(-1)^{m+1}\sqrt{\Delta_\rho}$, so
\begin{equation}
  \frac{\partial\rho_m}{\partial z_j}=(-1)^{m}\,\frac{g_j(\rho_m)}{\sqrt{\Delta_\rho}}.
  \label{eq:branch_sensitivity}
\end{equation}
The magnitude of the denominator $\sqrt{\Delta_\rho}$ is \emph{identical} for both roots. The normalized conditioning of the analytical root therefore cannot, by itself, distinguish the physical branch from the spurious one. This is an exact algebraic identity, and the branch asymmetry resides entirely in the numerator $g_j(\rho_m)$, evaluated at the two distinct root values. Collecting the two measurement components into $\mathbf{g}(\rho)=[\,g_1(\rho),\,g_2(\rho)\,]^{T}$, the per-root measurement sensitivity is $\varsigma_m=\|\nabla_{\mathbf z}\rho_m\|=\|\mathbf{g}(\rho_m)\|/\sqrt{\Delta_\rho}$, and the discriminating quantity is the ratio $\Gamma=\varsigma_{\mathrm{phys}}/\varsigma_{\mathrm{sp}}$. The common factor $\sqrt{\Delta_\rho}$ cancels, so $\Gamma=\|\mathbf{g}(\rho_{\mathrm{phys}})\|/\|\mathbf{g}(\rho_{\mathrm{sp}})\|$ depends only on geometry and the operating point.

That this numerator favors the physical branch rather than the spurious one is not incidental. It admits an exact closed-form criterion.

\begin{proposition}\label{prop:polarity}
Let the two feasible roots $\rho_{\mathrm{phys}},\rho_{\mathrm{sp}}$ be distinct and positive, so that $\Delta_\rho>0$. Since $\lambda_0=\boldsymbol{\alpha}^{T}\boldsymbol{\alpha}\geq0$, also $\lambda_2>0$. Then
\begin{equation}
  \Gamma>1 \iff (\rho_{\mathrm{phys}}-\rho_{\mathrm{sp}})\,Q>0,
  \label{eq:signlaw}
\end{equation}
where, with the coefficient gradients $\nabla_{\mathbf z}\lambda_k=[\,\partial_{z_1}\lambda_k,\,\partial_{z_2}\lambda_k\,]^{T}$ and the Vieta sum and product $S=-\lambda_1/\lambda_2$, $P=\lambda_0/\lambda_2$,
\begin{equation}
  \begin{aligned}
  Q={}&\|\nabla_{\mathbf z}\lambda_2\|^{2}\,S(S^{2}-2P)+2\,(\nabla_{\mathbf z}\lambda_2\!\cdot\!\nabla_{\mathbf z}\lambda_1)(S^{2}-P)\\
  &+\bigl(\|\nabla_{\mathbf z}\lambda_1\|^{2}+2\,\nabla_{\mathbf z}\lambda_2\!\cdot\!\nabla_{\mathbf z}\lambda_0\bigr)\,S+2\,\nabla_{\mathbf z}\lambda_1\!\cdot\!\nabla_{\mathbf z}\lambda_0 .
  \end{aligned}
  \label{eq:Q}
\end{equation}
\end{proposition}

\emph{Proof.} By \eqref{eq:implicit_sens} each $g_j$ is a quadratic in $\rho$, so $h(\rho)=\|\mathbf{g}(\rho)\|^{2}$ is a quartic whose coefficients are fixed by the geometry and operating point, with leading coefficient $\|\nabla_{\mathbf z}\lambda_2\|^{2}$ independent of the root index. The difference $h(\rho_{\mathrm{phys}})-h(\rho_{\mathrm{sp}})$ therefore factors as $(\rho_{\mathrm{phys}}-\rho_{\mathrm{sp}})$ times a polynomial symmetric in the two roots. Expressing that polynomial through the invariants $S$ and $P$ yields $h(\rho_{\mathrm{phys}})-h(\rho_{\mathrm{sp}})=(\rho_{\mathrm{phys}}-\rho_{\mathrm{sp}})\,Q$. Since $\Gamma>1\Leftrightarrow h(\rho_{\mathrm{phys}})>h(\rho_{\mathrm{sp}})$, \eqref{eq:signlaw} follows. \hfill$\square$

The leading term $\|\nabla_{\mathbf z}\lambda_2\|^{2}\,S(S^{2}-2P)=\|\nabla_{\mathbf z}\lambda_2\|^{2}\,S(\rho_{\mathrm{phys}}^{2}+\rho_{\mathrm{sp}}^{2})$ is non-negative, and $Q>0$ holds throughout the two-feasible-root interior (verified at all $2.3\times10^{6}$ sampled operating points). The more measurement-sensitive branch is thus the \emph{outer} (larger reference-range) root (Fig.~\ref{fig:polarity}(a)), and $\Gamma>1$ reduces to the physical root being the outer one. In the far field the leading term dominates and \eqref{eq:Q} gives $\Gamma\to(|\rho_{\mathrm{phys}}|/|\rho_{\mathrm{sp}}|)^{2}$. The physical root is the outer root except near the bifurcation curve $\Delta_\rho\to0$, where the roots coalesce, and the divergence locus $\lambda_2\to0$, where the leading coefficient of \eqref{eq:quadratic} vanishes \cite{kappagdop2026} so that one root tends to the finite limit $-\lambda_0/\lambda_1$ while the other diverges and inverts the polarity. The residual $\Gamma<1$ points form a thin set around these loci (Fig.~\ref{fig:polarity}(b)). This polarity is therefore a closed-form property of the analytical map, decided exactly by the sign of $Q$, and Proposition~\ref{prop:polarity} isolates the only regime, the neighborhood of the divergence locus, in which the single-epoch rule must be relinquished.

Two consequences fix the branch rule the asymmetry endorses. Because $Q>0$ throughout the interior, selecting the higher-sensitivity root is exactly the outer-root rule, choosing $\rho_{\mathrm{out}}=\max(\rho_1,\rho_2)$. That selection returns the physical branch if and only if the physical preimage is the outer one, $\rho_{\mathrm{phys}}>\rho_{\mathrm{sp}}$. On the thin inner-sheet minority where $\rho_{\mathrm{phys}}<\rho_{\mathrm{sp}}$, concentrated near the loci, the same rule returns the spurious branch. The data-only resolution is therefore to be read on this admissible interior, where the sensitivity polarity and the physical branch coincide. 

\begin{figure*}[t]
  \centering
  \includegraphics[width=\linewidth]{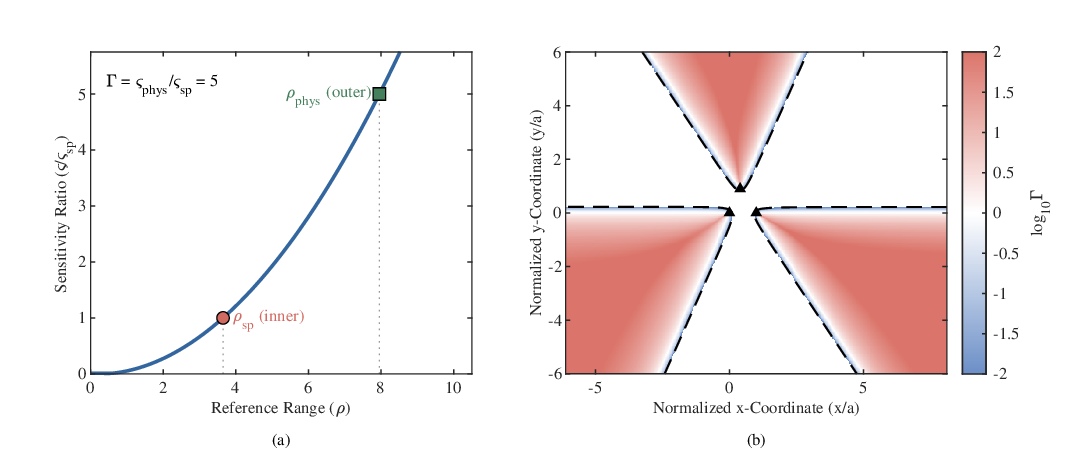}
  \caption{Sensitivity polarity (Proposition~\ref{prop:polarity}). (a) Per-root sensitivity $\varsigma(\rho)=\|\mathbf{g}(\rho)\|/\sqrt{\Delta_\rho}$ against reference range for the two roots of one operating point ($\Gamma=5$, $\varsigma$ normalized to the inner root). (b) $\log_{10}\Gamma$ over one geometry's source plane, $\Gamma>1$ in red and the $\Gamma<1$ minority in blue, with the bifurcation ($\Delta_\rho=0$, solid) and divergence ($\lambda_2=0$, dashed) loci and the receivers (triangles) marked.}
  \label{fig:polarity}
\end{figure*}

Consistent with this characterization, evaluating $\Gamma$ across a dimensionless atlas of $1558$ receiver geometries parameterized by the canonical-frame ratios $b/a$ and $c/a$ (Fig.~\ref{fig:atlas}) shows the physical-root dominance to be array-independent. The across-geometry median of $P(\Gamma>1)$ is $0.85$ and the median $\Gamma$ is $16$, with the sub-unity minority concentrated near the two loci, as Proposition~\ref{prop:polarity} predicts. The closed-form sensitivity \eqref{eq:branch_sensitivity} agrees with central finite differences to a median relative error of $1.8\times10^{-10}$ over $5834$ checks (supplementary material). The instantaneous fraction $P(\Gamma>1)$ is a conservative floor rather than the achievable accuracy. Because the median ratio $\Gamma=16$ is large, the smoothed statistic is dominated by the physical root within a few samples, so the integrated selector exceeds the per-point fraction and converges toward unity (Section~\ref{sec:validation}).

\begin{figure}[t]
  \centering
  \includegraphics[width=\linewidth]{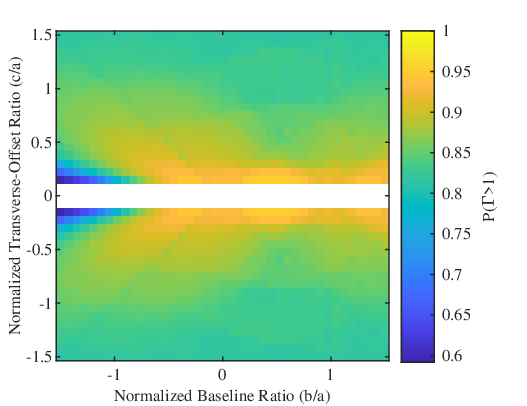}
  \caption{Branch-sensitivity asymmetry over a dimensionless atlas of $1558$ receiver geometries. The white band near $c/a=0$ marks excluded near-collinear arrays.}
  \label{fig:atlas}
\end{figure}

\section{CAUSAL SELECTOR AND SMOOTHNESS-INVERSION COROLLARY}
\label{sec:selector}
The polarity of Sec.~\ref{sec:sens} is a property of the noise-free operating point. This section constructs a causal statistic whose expectation recovers it from the measurement stream.
Let $\rho_1[n]$ and $\rho_2[n]$ be the formula-ordered roots of \eqref{eq:rho_roots}, taken with the positive and negative discriminant terms. For each, define the sampled root variation and its exponentially weighted moving average,
\begin{equation}
  \begin{aligned}
  d_m[n]&=\frac{|\rho_m[n]-\rho_m[n-1]|}{\max(\Delta t_n,\epsilon_t)},\\
  s_m[n]&=\eta\, s_m[n-1]+(1-\eta)d_m[n],
  \end{aligned}
  \label{eq:ema}
\end{equation}
with $0\leq\eta<1$ ($\eta=0.9$ here), where $\Delta t_n$ is the inter-sample interval and $\epsilon_t>0$ a small floor preventing division by zero. Select the more variable branch,
\begin{equation}
  \hat{m}[n]=\arg\max_{m\in\{1,2\}}s_m[n].
  \label{eq:branchchoice}
\end{equation}
To interpret $d_m[n]$, expand each root to first order about the noise-free
measurement. With $\boldsymbol{\nu}[n]$ the timing error in
\eqref{eq:measurement_vector} and $\nabla_{\mathbf z}\rho_m$ the closed-form root
gradient of Sec.~\ref{sec:sens},
\begin{equation}
  \rho_m[n]\approx\rho_m(\mathbf z[n]-\boldsymbol{\nu}[n])
  +(\nabla_{\mathbf z}\rho_m[n])^{T}\boldsymbol{\nu}[n].
  \label{eq:rho_expand}
\end{equation}
At a regular operating point the gradient varies slowly across two consecutive
samples, $\nabla_{\mathbf z}\rho_m[n]\approx\nabla_{\mathbf z}\rho_m[n-1]$, so
differencing \eqref{eq:rho_expand} between samples $n$ and $n-1$ splits the
step-to-step change into a noiseless and a noise part,
\begin{equation}
  \rho_m[n]-\rho_m[n-1]\approx\delta_m[n]
  +(\nabla_{\mathbf z}\rho_m)^{T}(\boldsymbol{\nu}[n]-\boldsymbol{\nu}[n-1]),
  \label{eq:rho_diff}
\end{equation}
where $\delta_m[n]=\rho_m(\mathbf z[n]-\boldsymbol{\nu}[n])
-\rho_m(\mathbf z[n-1]-\boldsymbol{\nu}[n-1])$ is the noiseless displacement of
root $m$. Dividing by $\max(\Delta t_n,\epsilon_t)$ and taking the magnitude,
\begin{equation}
  d_m[n]\approx\frac{\bigl|\delta_m[n]+(\nabla_{\mathbf z}\rho_m)^{T}
  (\boldsymbol{\nu}[n]-\boldsymbol{\nu}[n-1])\bigr|}{\max(\Delta t_n,\epsilon_t)},
  \label{eq:dm_decomp}
\end{equation}
so the sampled variability mixes a source-motion term $\delta_m[n]$ with a noise
term whose magnitude is set by the root sensitivity $\nabla_{\mathbf z}\rho_m$.

Model the timing error as zero mean, independent from sample to sample, and
Gaussian with covariance $\sigma_\nu^{2}\mathbf{I}_2$, where $\sigma_\nu$ is the
timing standard deviation expressed in range units (Assumption~1), that is,
\begin{equation}
  \mathbb{E}\{\boldsymbol{\nu}[n]\}=\mathbf{0},\qquad
  \mathbb{E}\{\boldsymbol{\nu}[n]\boldsymbol{\nu}[n]^{T}\}
  =\sigma_\nu^{2}\mathbf{I}_2.
  \label{eq:noise_model}
\end{equation}
The first difference $\boldsymbol{\nu}[n]-\boldsymbol{\nu}[n-1]$ is then zero mean
with covariance $2\sigma_\nu^{2}\mathbf{I}_2$, so the scalar noise term in
\eqref{eq:dm_decomp} is zero mean with variance
\begin{equation}
  \mathrm{Var}\bigl((\nabla_{\mathbf z}\rho_m)^{T}
  (\boldsymbol{\nu}[n]-\boldsymbol{\nu}[n-1])\bigr)
  =2\sigma_\nu^{2}\|\nabla_{\mathbf z}\rho_m\|^{2}.
  \label{eq:noise_var}
\end{equation}
When this noise term dominates $|\delta_m[n]|$ (the noise-dominance condition),
$d_m[n]$ is the normalized magnitude of a single zero-mean Gaussian.

\begin{proposition}\label{prop:expvar}
At a regular two-feasible-root operating point away from the bifurcation and
divergence loci, under Assumption~1 and the noise-dominance condition, to first
order
\begin{equation}
  \mathbb{E}\{d_m[n]\}\approx
  \frac{2\sigma_\nu}{\sqrt{\pi}\,\max(\Delta t_n,\epsilon_t)}\,\varsigma_m,
  \label{eq:expected_dm}
\end{equation}
all quantities evaluated at sample $n$. Under the same conditions,
$\mathbb{E}\{d_1[n]\}>\mathbb{E}\{d_2[n]\}\iff\varsigma_1>\varsigma_2
\iff\|\mathbf{g}(\rho_1)\|>\|\mathbf{g}(\rho_2)\|$, since both roots share
$\Delta t_n$ and $\sqrt{\Delta_\rho}$.
\end{proposition}

\emph{Proof.} Under the noise-dominance condition \eqref{eq:dm_decomp} reduces
to $d_m[n]\approx|(\nabla_{\mathbf z}\rho_m)^{T}(\boldsymbol{\nu}[n]
-\boldsymbol{\nu}[n-1])|/\max(\Delta t_n,\epsilon_t)$, the normalized magnitude of
the zero-mean Gaussian whose variance is \eqref{eq:noise_var}. For
$X\sim\mathcal{N}(0,\sigma_X^{2})$ one has $\mathbb{E}|X|=\sqrt{2/\pi}\,\sigma_X$,
so
\begin{equation}
  \begin{aligned}
  \mathbb{E}\{d_m[n]\}
  &\approx\frac{1}{\max(\Delta t_n,\epsilon_t)}\sqrt{\tfrac{2}{\pi}}\,
  \sqrt{2\sigma_\nu^{2}\|\nabla_{\mathbf z}\rho_m\|^{2}}\\
  &=\frac{2\sigma_\nu}{\sqrt{\pi}\,\max(\Delta t_n,\epsilon_t)}\,
  \|\nabla_{\mathbf z}\rho_m\|.
  \end{aligned}
  \label{eq:proof_step}
\end{equation}
Substituting $\|\nabla_{\mathbf z}\rho_m\|=\varsigma_m$ from Sec.~\ref{sec:sens}
gives \eqref{eq:expected_dm}. The biconditional follows because the prefactor and
$\sqrt{\Delta_\rho}$ are common to both roots.\hfill$\square$

On the admissible interior of Sec.~\ref{sec:sens}, where the physical root is
the outer one and $Q>0$ makes the outer root the more sensitive, this expected
ordering therefore points at the physical branch. Like the polarity itself, the
approximation degrades near the bifurcation and divergence loci, where the
first-order expansion no longer holds. Numerically, over the admissible interior
the two sides of \eqref{eq:expected_dm} agree to within five percent and the
expected and realized variability orderings coincide at $99.5\%$ of sampled
operating points. When the
statistics of $d_m[n]$ vary slowly over the averaging window the exponential
moving average preserves the mean,
$\mathbb{E}\{s_m[n]\}\approx\mathbb{E}\{d_m[n]\}$, so \eqref{eq:ema} carries the
expected ordering to $s_m[n]$ while reducing its variance. For independent
increments
\begin{equation}
  \mathrm{Var}(s_m[n])\approx\frac{1-\eta}{1+\eta}\,\mathrm{Var}(d_m[n]),
  \label{eq:ema_var}
\end{equation}
a factor of $19$ at $\eta=0.9$. Because consecutive increments share
$\boldsymbol{\nu}[n]$, so that the first differences form a first-order moving
average, the achieved reduction is smaller and \eqref{eq:ema_var} is an upper
bound.
This variance reduction is what smoothing buys
(Sec.~\ref{sec:validation}). When the noiseless step $\delta_m[n]$ is comparable
to the noise term in \eqref{eq:dm_decomp}, the expected ordering of
Proposition~\ref{prop:expvar} is no longer guaranteed. The noiseless end of the
validation curves in Sec.~\ref{sec:validation} shows this regime. The rule is causal and uses constant memory per receiver triplet. It operates on the sign-ordered traces directly. This ordering is continuous through the discriminant zero $\Delta_\rho\to0$, where both roots coincide and $d_m[n]$ stays finite, whereas a magnitude-sort labeling would inject an artificial one-sample jump whenever the roots cross in magnitude.

Propositions~\ref{prop:polarity} and~\ref{prop:expvar} have an immediate corollary that inverts the conventional intuition. On the admissible interior the physical root carries the larger expected variability, so under timing noise it traces the higher-variance arc while the spurious root traces a tighter, lower-variance one. The single per-root variance therefore carries two opposite consequences. Selecting the \emph{larger} variability, as in \eqref{eq:branchchoice}, locks onto the physical branch, whereas any rule that prefers the \emph{smoother} or more temporally consistent candidate, including the trajectory-energy benchmark of \cite{samizadehnikoo2019} and a constant-velocity continuity rule, must select the spurious one once noise dominates. The two read the same variance signal with opposite polarity, which is why the temporal-variability selector improves with noise exactly where smoothness-based disambiguation collapses.

\section{VALIDATION}
\label{sec:validation}
The selector was evaluated against two smoothness-based references by Monte Carlo over randomized two-feasible-root crossing geometries, with $1000$ trajectories at each timing-noise level $\sigma_t\in\{0,5,10,20,50\}$~ns. The references are the trajectory-energy method of \cite{samizadehnikoo2019}, which selects the candidate trajectory of minimum total absolute energy-change rate, and a constant-velocity continuity selector, which selects the candidate most consistent with constant-velocity motion. Fig.~\ref{fig:mc} shows the interior branch-selection accuracy. TVSD rises from $0.82$ at $0$~ns and saturates near $0.98$ once $\sigma_t\geq5$~ns. This is the profile Proposition~\ref{prop:expvar} predicts. At $0$~ns the noise-dominance condition fails and $d_m[n]$ reflects source motion rather than sensitivity, while under noise the expected ordering takes hold. Both smoothness-based selectors are accurate without noise but collapse far below the chance line once timing noise is present, the trajectory-energy method to near $0.01$ and the constant-velocity selector from $0.92$ to near $0.03$, exactly as the corollary predicts, since the smoothest candidate is the spurious one once noise dominates the interior.

\begin{figure}[t]
  \centering
  \includegraphics[width=\linewidth]{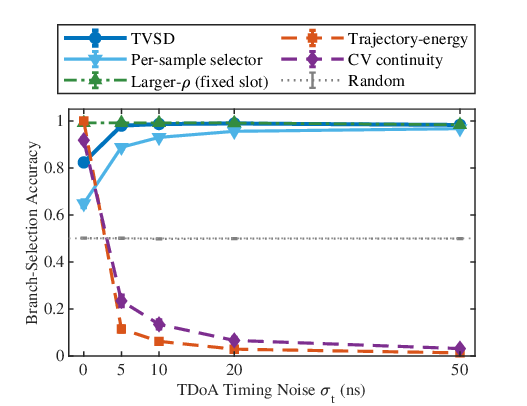}
  \caption{Monte Carlo branch-selection accuracy versus TDoA timing noise in the two-feasible-root interior (mean and 95\% confidence interval, $1000$ trajectories per level). Series: TVSD; the larger-$\rho$ (outer-root) rule; a memoryless per-sample selector; and two smoothness-based selectors, the trajectory-energy method \cite{samizadehnikoo2019} and a constant-velocity continuity rule. Random marks the $0.5$ chance line.}
  \label{fig:mc}
\end{figure}

Two elementary branch rules invite comparison beyond the smoothness references. The first selects the larger reference range $\rho_{\mathrm{out}}$. By Proposition~\ref{prop:polarity} this outer-root rule is exactly the fixed analytical-slot (oracle) choice in this interior, since $Q>0$ makes the $+\sqrt{\Delta_\rho}$ slot the larger root, and it sits near $0.99$. The second selects, at each sample, the root of larger instantaneous sensitivity $\|\mathbf{g}(\rho_m)\|$. Because $Q>0$ forces the more sensitive root to be the outer one, this rule is identical to the outer-root rule and traces the same curve. This coincidence is the paper's point rather than a baseline to be beaten. The sensitivity asymmetry is what makes the simple outer-root rule correct on the admissible interior, and the atlas of Fig.~\ref{fig:atlas} shows the coincidence holds across geometries, so the rule is theory-justified rather than tuned to one array. TVSD realizes the same selection causally and with constant memory, as an online statistic whose expectation is the per-root sensitivity (Proposition~\ref{prop:expvar}), and matches the outer-root rule near 0.98. A per-sample selector that reads the root variability from a single time step, with no memory, is markedly less accurate, rising from $0.65$ to $0.97$ yet remaining below TVSD at every noise level (Fig.~\ref{fig:mc}). The gap is consistent with the variance reduction of the smoothed statistic (Sec.~\ref{sec:selector}), so the temporal integration is essential to the selector. The analytical solver, the selector, and the scripts that generate the Monte Carlo and atlas results are fully reproducible and will be released publicly upon publication. This article uses no proprietary data.

\section{CONCLUSION}

This article gives a closed-form characterization of the branch-selection ambiguity in minimal-array analytical TDoA, rather than merely observing it. An explicit sign quantity $Q$ (Proposition~\ref{prop:polarity}) decides which branch is the more measurement-sensitive, and over the feasible interior this is the outer, physical root away from the bifurcation and divergence loci, which makes the simple outer-root rule a theorem-backed selector with a known failure map. The causal, constant-memory selector then reads that polarity directly from the measurement stream, its variability statistic being proportional, in expectation and under a noise-dominance condition, to the per-root sensitivity (Proposition~\ref{prop:expvar}). It realizes this selection causally and online, matching the outer-root rule near 0.98 under timing noise. The same sensitivity analysis thus both supplies the discriminant and explains, through the smoothness-inversion corollary, why prior smoothness- and continuity-based rules mislabel the physical branch exactly where noise makes the problem hardest. Three directions remain open. One embeds the selector in a downstream tracker that carries residual ambiguity as a soft hypothesis rather than a hard decision. Another restores selection near the divergence locus, where the polarity inverts, through a multi-epoch continuity rule. The last is a closed-form proof of the positivity $Q>0$, here established numerically over the feasible interior, which would render the characterization unconditional.

\end{document}